\documentclass[12pt]{article}

\usepackage[hmargin=1.3in,vmargin=1.3in]{geometry}
\usepackage{bbding}
\usepackage{mathrsfs}
\usepackage{}
\usepackage{amsfonts}
\usepackage{multirow}
\usepackage{amsfonts,amssymb,amsmath,amsthm,bm}
\usepackage[colorlinks,
            linkcolor=blue,
            anchorcolor=blue,
            citecolor=blue]{hyperref}

\newtheorem{theorem}{Theorem}

\newtheorem{remark}{Remark}

\newtheorem{lemma}{Lemma}

\begin{document}

\title{Euclidean and Hermitian Hulls of MDS Codes and Their Applications to EAQECCs\thanks{The research of W. Fang and F.-W. Fu  is supported in part by the National
Natural Science Foundation of China (Grant Nos. 61971243, 61571243
and 61771273), the Nankai Zhide Foundation, and the Research Fund of
PCL Future Regional Network Facilities for Large-Scale Experiments and
Applications (PCL2018KP001). The research of L. Li and S. Zhu is supported in part by the National Natural Science Foundation of China (Grant No. 61772168).}}
\author{\small  Weijun Fang\thanks{Corresponding Author} $^{,1,2,3}$, Fang-Wei Fu$^{3,4}$, Lanqiang Li$^{5}$, Shixin Zhu$^{5}$\\
\scriptsize $^1$  \emph{Shenzhen International Graduate School, Tsinghua
University, Shenzhen, P.R.China}\\
\scriptsize $^2$  \emph{ PCL Research Center of
Networks and Communications, Peng Cheng Laboratory, Shenzhen,
P.R.China}\\
\scriptsize $^3$  \emph{Chern Institute of Mathematics and LPMC, Nankai University, Tianjin, P.R.China}\\
\scriptsize $^4$  \emph{Tianjin Key Laboratory of Network and Data Security Technology,
Nankai University, Tianjin, P.R.China}\\
\scriptsize $^5$  \emph{School of Mathematics, Hefei University of
			Technology, Hefei, Anhui, P.R.China}\\
\scriptsize \emph{E-mail}: nankaifwj@163.com, fwfu@nankai.edu.cn, lilanqiang716@126.com, zhushixinmath@hfut.edu.cn}
\date{}
\maketitle
\thispagestyle{empty}
\begin{abstract}
In this paper, we construct several classes of \emph{maximum distance separable} (MDS) codes via generalized Reed-Solomon (GRS) codes and extended
GRS codes, where we can determine the dimensions of their Euclidean hulls or Hermitian hulls. It turns out that the dimensions of Euclidean hulls or Hermitian hulls of the codes in our constructions can take all or almost all possible values. As a consequence, we can apply our results to entanglement-assisted quantum error-correcting codes (EAQECCs) and obtain several new families of MDS EAQECCs with flexible parameters. The required number of maximally entangled states of these MDS EAQECCs can take all or almost all possible values. Moreover, several new classes of $q$-ary MDS EAQECCs of length $n >q+1$ are also obtained.
\end{abstract}


\maketitle

\section{Introduction}

Let $\mathcal{C}$ be a linear code over a finite field, and let $\mathcal{C^{\perp}}$ be the dual code of $\mathcal{C}$ with respect to certain inner product, such as Euclidean inner product and Hermitian inner product. The hull of $\mathcal{C}$ is just defined as the intersection $Hull(\mathcal{C})=\mathcal{C}\bigcap \mathcal{C^{\perp}}$. Some research topics in coding theory are closely related to the properties of the hull of a linear code. One interesting problem in coding theory is that to decide whether
two matrices generate equivalent linear codes and compute the permutation of two given equivalent linear codes (\cite{L79,PR97}). In \cite{L82,L91,S97,S00}, the authors provided some algorithms for these computations whose complexity is determined by the dimension of the Euclidean hull of linear codes. Some properties of the hull of cyclic codes and negacyclic codes were also studied in \cite{S03} and \cite{SJLU15}.

It is worth mentioning that two special cases of the hulls of linear codes are of much interest. One is that $Hull(\mathcal{C}) = \{0\}$, in which $\mathcal{C}$ is called a linear complementary dual (LCD) code. In \cite{M92}, Massey first introduced this class of codes and proved that there exist asymptotically good LCD codes. A practical application of binary LCDs against side-channel attacks (SCAs) and fault injection
attacks (FIAs) was investigated by Carlet \emph{et al.} in \cite{BCCG14} and \cite{CG14}. The study of LCD codes is thus
becoming a hot research topic in coding theory (\cite{J17, L18,LDL17,CL18,CMTQP18}). A surprising result was given in \cite{CMTQP18}, which proved that any linear code over
$\mathbb{F}_{q}$ $(q > 3)$ is equivalent to a Euclidean LCD code and any linear
code over $\mathbb{F}_{q^{2}}$ $(q > 2)$ is equivalent to a Hermitian LCD code. The other case is that $Hull(\mathcal{C}) = \mathcal{C}$ (resp. $\mathcal{C}^{\perp}$). Such codes are called self-orthogonal (resp. dual containing) codes. Calderbank \emph{et al.} \cite{CS96} and  Steane \cite{S96} presented an effective mathematical method to construct good quantum stabilizer  codes
from classical self-orthogonal codes (or dual containing codes) over finite fields. Since then, several families of quantum stabilizer codes have been constructed by classical linear codes with certain self-orthogonality.

In \cite{BDH06}, Brun \emph{et al.} introduced entanglement-assisted quantum error-correcting codes (EAQECCs), which include the standard quantum stabilizer codes as a special case. They showed that if pre-shared entanglement between the encoder and decoder is available, the EAQECCs can be constructed via classical linear codes without self-orthogonality. Moreover, an EAQECC is MDS if and only if the corresponding classical linear code is MDS. However, it is not easy to determine the number of shared pairs that required to construct an EAQECC. Several classes of MDS EAQECCs had been constructed with some fixed numbers of shared pairs (\cite{QZ18, FCX16, LLGML18, LZLK19, CZK18, LLLM18, K19, LC19}).  Guenda \emph{et al.} \cite{GJG18} provided the relation between this number and the dimension of the hull of classical linear codes. Therefore, it is important to study the hull of linear codes, in particular for MDS codes. Very recently, Luo \emph{et al.} \cite{LCC} presented several classes of GRS and extended GRS codes with Euclidean hulls of arbitrary dimensions and constructed some families of $q$-ary MDS EAQECCs with length $n \leq q+1$. In \cite{GGJT18}, Guenda \emph{et al.} investigated the $\ell$-intersection pair of linear codes, which is a generalization of linear complementary pairs of codes. And then, they completely determined the $q$-ary MDS EAQECCs of length $n \leq q+1$ for all possible parameters.

In this paper, we construct several MDS codes by utilizing GRS codes and extended GRS codes, and determine the dimensions of their Euclidean or Hermitian hulls. More precisely, we firstly give some new classes of MDS codes with Euclidean hulls of arbitrary dimensions whose parameters are not covered in \cite{LCC}. Secondly, several new classes of MDS codes with Hermitian hulls of arbitrary dimensions are presented.  Finally, we apply these results to construct new MDS EAQECCs. In particular, we provide a different and simpler method to construct the $q$-ary MDS EAQECCs of length $n \leq q$ for all possible parameters, which were first obtained in \cite{GGJT18}. Furthermore, we obtain several new classes of $q$-ary MDS EAQECCs with length larger than $q+1$ and the required number of maximally entangled states can take all or almost all possible values.

The rest of this paper is organized as follows. In Section 2, we briefly recall some basic notions and properties of GRS codes and extended GRS codes. In Section 3, we present our constructions
of MDS codes with Euclidean or Hermitian hulls of arbitrary dimensions. Several classes of MDS EAQECCs are obtained in Section 4. We conclude this paper in Section 5.

\section{Preliminaries}\label{sec2}
In this section, we introduce the basic notions of Euclidean and Hermitian hulls of a linear code, and provide some related properties of GRS codes and extended GRS codes.

Throughout this paper, we always assume that $p$ is a prime and $q=p^{m}$, where $m$ is a positive integer. Let $\mathbb{F}_{q}$ be the finite field with $q$ elements and $\mathbb{F}^{*}_{q}=\mathbb{F}_{q}\setminus \{0\}$. For any two vectors $\textbf{u}=(u_{1}, u_{2}, \ldots, u_{n}), \textbf{v}=(v_{1}, v_{2}, \ldots, v_{n}) \in \mathbb{F}_{q}^{n},$ the Euclidean inner product $\textbf{u}$ and $\textbf{v}$ is defined by
\[ \langle \textbf{u}, \textbf{v}\rangle_{E} \triangleq\sum_{i=1}^{n}u_{i}v_{i}. \]

Let $\mathcal{C}$ be an $\mathbb{F}_{q}$-linear code of length $n$, the Euclidean dual code of $\mathcal{C}$ is defined as
\[ \mathcal{C}^{\perp_{E}} \triangleq\{\textbf{u} \in \mathbb{F}_{q}^{n} :\langle \textbf{u}, \textbf{v}\rangle_{E}=0 \textnormal{ for all } \textbf{v} \in \mathcal{C} \}. \]

Similarly, for any two vectors $\textbf{u}, \textbf{v} \in \mathbb{F}_{q^{2}}^{n}$, the Hermitian inner product of $\textbf{u}$ and $\textbf{v}$ is defined as
\[ \langle \textbf{u}, \textbf{v}\rangle_{H} \triangleq\sum_{i=1}^{n}u_{i}v_{i}^{q}. \]

Let $\mathcal{C}$ be an $\mathbb{F}_{q^{2}}$-linear code of length $n$. We can similarly define the Hermitian dual code of $\mathcal{C}$ as follows:
\[ \mathcal{C}^{\perp_{H}} \triangleq\{\textbf{u} \in \mathbb{F}_{q^{2}}^{n} :\langle \textbf{u}, \textbf{v}\rangle_{H}=0 \textnormal{ for all } \textbf{v} \in \mathcal{C} \}.\]

It is worth mentioning that the base field should be $\mathbb{F}_{q^{2}}$ when we consider the Hermitian case in this paper.

The Euclidean hull (resp. Hermitian hull) of $\mathcal{C}$ is just the intersection $\mathcal{C} \bigcap \mathcal{C}^{\perp_{E}}$ (resp. $\mathcal{C} \bigcap \mathcal{C}^{\perp_{H}}$), which we denote by $Hull_{E}(\mathcal{C})$ (resp. $Hull_{H}(\mathcal{C})$). It is obvious that $Hull_{E}(\mathcal{C})=Hull_{E}(\mathcal{C}^{\perp_{E}})$ (resp. $Hull_{H}(\mathcal{C})=Hull_{H}(\mathcal{C}^{\perp_{H}})$). If $Hull_{E}(\mathcal{C})=0$ (resp. $Hull_{H}(\mathcal{C})=0$), $\mathcal{C}$ is called a Euclidean LCD (resp. Hermitian LCD) code. If $Hull_{E}(\mathcal{C})=\mathcal{C}$ (resp. $Hull_{H}(\mathcal{C})=\mathcal{C}$), $\mathcal{C}$ is called a self-orthogonal (resp. Hermitian self-orthogonal) code.

In general, it is not an easy task to determine the dimension of the Euclidean (or Hermitian) hull of a linear code. In this paper, we will give several constructions of MDS codes and determine the dimensions of their Euclidean (or Hermitian) hulls. Let's recall some basic notions of GRS codes and extended GRS codes. Let $a_{1}, a_{2}, \ldots, a_{n}$ be $n$ distinct elements of $\mathbb{F}_{q}$ and $v_{1}, v_{2}, \ldots, v_{n}$ be $n$ nonzero elements of $\mathbb{F}_{q}$. Put $\textbf{a}= (a_{1}, \ldots, a_{n})$ and $\textbf{v}=(v_{1}, \ldots, v_{n})$. The generalized Reed-Solomon (GRS for short) code over $\mathbb{F}_{q}$ associated to $\textbf{a}$ and $\textbf{v}$ is defined as follows:
\begin{eqnarray*}
GRS_{k}(\textbf{a}, \textbf{v}) &\triangleq& \{(v_{1}f(a_{1}), \ldots, v_{n}f(a_{n})) \\
                                & &: f(x) \in \mathbb{F}_{q}[x], \deg(f(x)) \leq k-1 \}.
\end{eqnarray*}
It is well known that the code $GRS_{k}(\textbf{a}, \textbf{v})$ is an $[n, k, n - k + 1]$-MDS code.

The extended GRS code $GRS_{k}(\textbf{a}, \textbf{v},\infty)$ associated to $\textbf{a}$ and $\textbf{v}$ is defined by
\begin{eqnarray*}
  GRS_{k}(\textbf{a}, \textbf{v},\infty)&\triangleq& \{(v_{1}f(a_{1}), \ldots, v_{n}f(a_{n}),f_{k-1}) \\
    & & :f(x) \in \mathbb{F}_{q}[x], \deg(f(x)) \leq k-1 \},
\end{eqnarray*}
where $f_{k-1}$ stands for the coefficient of $x^{k-1}$ in $f(x)$. It is easy to show that $GRS_{k}(\textbf{a}, \textbf{v},\infty)$ is an $[n+1, k, n - k + 2]$-MDS code (see \cite[Theorem 5.3.4]{HP03}).  For $1 \leq i \leq n$, we denote
\begin{equation}\label{1}
  u_{i}\triangleq\prod_{1 \leq j \leq n, j \neq i}(a_{i}-a_{j})^{-1},
\end{equation}
which will be used frequently in this paper.

In \cite{CL18}, the authors presented a sufficient and necessary condition under which a codeword $\textbf{c}$ of $GRS_{k}(\textbf{a}, \textbf{v})$ (resp. $GRS_{k}(\textbf{a}, \textbf{v},\infty)$) is contained in its dual code $GRS_{k}(\textbf{a}, \textbf{v})^{\perp_{E}}$ (resp. $GRS_{k}(\textbf{a}, \textbf{v},\infty)^{\perp_{E}}$).

\begin{lemma} (\cite[Lemma III.1]{CL18})\label{lem2.1}
A codeword $\textbf{c}=(v_{1}f(a_{1}), v_{2}f(a_{2}), \ldots, v_{n}f(a_{n}))$ of $GRS_{k}(\textbf{a}, \textbf{v})$ is contained in $GRS_{k}(\textbf{a}, \textbf{v})^{\perp_{E}}$ if and only if there exists a polynomial $g(x)\in \mathbb{F}_{q}[x]$ with $\deg(g(x)) \leq n-k-1$, such that
\begin{equation*}
  \begin{split}
      & (v_{1}^{2}f(a_{1}), v_{2}^{2}f(a_{2}), \ldots, v_{n}^{2}f(a_{n})) \\
       & =(u_{1}g(a_{1}), u_{2}g(a_{2}),\ldots, u_{n}g(a_{n})).
   \end{split}
\end{equation*}
\end{lemma}

\begin{lemma}(\cite[Lemma III.2]{CL18}\footnote{Indeed, \cite[Lemma III.2]{CL18} only considered the case of $n=q$. It can similarly prove that the lemma holds for general $n \leq q$.})\label{lem2.2}
A codeword $\textbf{c}=(v_{1}f(a_{1}), \ldots, v_{n}f(a_{n}), f_{k-1})$ of $GRS_{k}(\textbf{a}, \textbf{v}, \infty)$ is contained in $GRS_{k}(\textbf{a}, \textbf{v}, \infty)^{\perp_{E}}$ if and only if there exists a polynomial $g(x) \in \mathbb{F}_{q}[x]$ with $\deg(g(x)) \leq n-k$, such that
\begin{equation*}
   \begin{split}
      & (v_{1}^{2}f(a_{1}), v_{2}^{2}f(a_{2}), \ldots, v_{n}^{2}f(a_{n}), f_{k-1}) \\
       & =(u_{1}g(a_{1}), u_{2}g(a_{2}), \ldots, u_{n}g(a_{n}),-g_{n-k}).
   \end{split}
\end{equation*}
\end{lemma}

Similar results for the Hermitian case were obtained in \cite{FF18}.

\begin{lemma}(\cite[Lemma 6]{FF18})\label{lem2.3}
A codeword $\textbf{c}=(v_{1}f(a_{1}), v_{2}f(a_{2}), \ldots, v_{n}f(a_{n}))$ of $GRS_{k}(\textbf{a}, \textbf{v})$ is contained in $GRS_{k}(\textbf{a}, \textbf{v})^{\perp_{H}}$ if and only if there exists a polynomial $g(x)\in \mathbb{F}_{q^{2}}[x]$ with $\deg(g(x)) \leq n-k-1$, such that
\begin{equation*}
  \begin{split}
      & (v_{1}^{q+1}f^{q}(a_{1}), v_{2}^{q+1}f^{q}(a_{2}), \ldots, v_{n}^{q+1}f^{q}(a_{n})) \\
       & =(u_{1}g(a_{1}), u_{2}g(a_{2}),\ldots, u_{n}g(a_{n})).
   \end{split}
\end{equation*}
\end{lemma}

\begin{lemma}(\cite[Lemma 7]{FF18})\label{lem2.4}
A codeword $\textbf{c}=(v_{1}f(a_{1}), \ldots, v_{n}f(a_{n}), f_{k-1})$ of $GRS_{k}(\textbf{a}, \textbf{v}, \infty)$ is contained in $GRS_{k}(\textbf{a}, \textbf{v}, \infty)^{\perp_{H}}$ if and only if there exists a polynomial $g(x)\in \mathbb{F}_{q^{2}}[x]$ with $\deg(g(x)) \leq n-k$, such that
\begin{eqnarray*}
\begin{split}
   & (v_{1}^{q+1}f^{q}(a_{1}), v_{2}^{q+1}f^{q}(a_{2}), \ldots, v_{n}^{q+1}f^{q}(a_{n}), f^{q}_{k-1}) \\
    & =(u_{1}g(a_{1}), u_{2}g(a_{2}), \ldots, u_{n}g(a_{n}), -g_{n-k}).
\end{split}
\end{eqnarray*}
\end{lemma}

Lemmas \ref{lem2.1}-\ref{lem2.4} will play important roles in calculating the dimension of the hull of the MDS codes constructed in Section 3.

\section{Constructions}
In this section, we will provide several families of GRS codes and extended GRS codes with Euclidean hulls or Hermitian hulls of arbitrary dimensions. The main idea of our constructions is to choose $n$ suitable distinct elements $a_{1}, a_{2}, \ldots, a_{n} \in \mathbb{F}_{q}$ (or $\mathbb{F}_{q^{2}}$) such that each value of $u_{i}$ defined by Eq. \eqref{1} can be easily calculated.

\subsection{MDS Codes with Euclidean Hulls of Arbitrary Dimensions}
In this subsection, we will provide some constructions of MDS codes with Euclidean hulls of arbitrary dimensions. Since $Hull_{E}(\mathcal{C})=Hull_{E}(\mathcal{C}^{\perp_{E}})$, we always assume that the dimension $k$ is less than or equal to half of the code length in our constructions.

The first construction is based on an additive subgroup of $\mathbb{F}_{q}$ and its cosets. Let $q=p^{m}$ and $r=p^{e}$, where $e \mid m$. Then $\mathbb{F}_{q}$ can be seen as a linear space over $\mathbb{F}_{r}$ of dimension $\frac{m}{e}$. Suppose $1 \leq t \leq r$ and $1 \leq z \leq \frac{m}{e}-1$, let $H$ be an $\mathbb{F}_{r}$-subspace of $\mathbb{F}_{q}$ (or $\mathbb{F}_{q^{2}}$ in the proof of Theorem \ref{thm3.6}) of dimension $z$. Choose $\eta \in \mathbb{F}_{q} \setminus H$ (or $\mathbb{F}_{q^{2}} \setminus H$). Label the elements of $\mathbb{F}_{r}$ as $\beta_{1}=0, \beta_{2},\ldots, \beta_{r}$. For $1 \leq j \leq t$, define
                     \[H_{j}:=H+\beta_{j}\eta:=\{h+\beta_{j}\eta \mid h \in H\}.\]
Let $n=tr^{z}$ and
\begin{equation}\label{2}
  \bigcup_{j=1}^{t}H_{j}:=\{a_{1}, a_{2}, \ldots, a_{n}\}.
\end{equation}

For $1 \leq i \leq n$, $u_{i}$ is defined as in \eqref{1}. Similar to \cite[Lemmas 8 and 9]{FF18}, the value of $u_{i}$ is given as follows.
\begin{lemma}\label{lem3.1}
For a given $1 \leq i \leq n$, suppose $a_{i} \in H_{b}$ for some $1 \leq b \leq t$. Then we have
\[u_{i}=(\prod_{h \in H, h \neq 0} h^{-1})(\prod_{g \in H}(\eta-g))^{1-t}\left(\prod_{1 \leq j \leq t, j\neq b}(\beta_{b}-\beta_{j})^{-1}\right). \]
In particular, let $\varepsilon=(\prod\limits_{h \in H, h \neq 0} h)(\prod\limits_{g \in H}(\eta-g))^{t-1}$, then
\[\varepsilon u_{i} \in \mathbb{F}_{r}^{*}.\]
\end{lemma}

\begin{proof}
Suppose $a_{i}=\xi+\beta_{b}\eta$, for some $\xi \in H$. Then
\begin{eqnarray*}
  u_{i} &=& \prod_{1 \leq j \leq n, j \neq i}(a_{i}-a_{j})^{-1} \\
   &=& \prod_{h_{b} \in H_{b}, h_{b} \neq a_{i}}(a_{i}-h_{b})^{-1}\prod_{1\leq j \leq t, j \neq b}\prod_{h_{j}\in H_{j}}(a_{i}-h_{j})^{-1}.
\end{eqnarray*}
Note that
\begin{eqnarray*}
  \prod_{h_{b} \in H_{b}, h_{b} \neq a_{i}}(a_{i}-h_{b}) &=& \prod_{\gamma \in H, \gamma \neq \xi}(\xi+\beta_{b}\eta-(\gamma+\beta_{b}\eta)) \\
    &=& \prod_{h \in H, h \neq 0} h,
\end{eqnarray*}
and for $j \neq b$,
\begin{eqnarray*}
  \prod_{h_{j}\in H_{j}}(a_{i}-h_{j}) &=& \prod_{\gamma \in H}(\xi+\beta_{b}\eta-(\gamma+\beta_{j}\eta)) \\
    &=& \prod_{g \in H}((\beta_{b}-\beta_{j})\eta-g) \\
    &=& (\beta_{b}-\beta_{j})\prod_{g \in H}(\eta-g).
\end{eqnarray*}
The last equality holds since $\beta_{b},\beta_{j} \in \mathbb{F}_{r}$. The lemma is proved.
\end{proof}
Before giving our constructions, we need the following simple lemma.

\begin{lemma}\label{lem3.2}
Let $F$ be a finite field and $A \subsetneqq F$. Then, for any integer $\ell \geq 0$ , there exists a monic polynomial $\pi(x) \in F[x]$ of degree $\ell$ such that $\pi(a) \neq 0 $ for all $a \in A$.
\end{lemma}
\begin{proof}
For $\ell=0$, let $\pi(x)=1$; For $\ell=1$, let $\pi(x)=x-\delta$, where $\delta \in F\setminus A;$ For $\ell \geq 2$, the conclusion follows from \cite[Lemma 12]{FF18}.
\end{proof}

\begin{theorem}\label{thm3.3}
Let $q=p^{m}>2$ and $r=p^{e}$, where $e \mid m$. Suppose  $\frac{m}{e}$ is even. Let $n=tr^{z}$, where $1 \leq t \leq r$ and $1 \leq z \leq \frac{m}{e}-1$.
\begin{description}
  \item[(i)] For any $1 \leq k \leq \lfloor\frac{n}{2}\rfloor$ and $0 \leq \ell \leq k$, then there exists a $q$-ary $[n,k]$-MDS code $\mathcal{C}$ with $\dim(Hull_{E}(\mathcal{C}))=\ell$.
  \item[(ii)] If $n$ is even, then for any $1 \leq k \leq \frac{n}{2}$ and $0 \leq \ell \leq k-1$, there exists a $q$-ary $[n+1,k]$-MDS code $\mathcal{C}$ with $\dim(Hull_{E}(\mathcal{C}))=\ell$.
  \item[(iii)] If $n$ is odd and $n <q$, then for any $1 \leq k \leq \frac{n+1}{2}$ and $0 \leq \ell \leq k$, there exists a $q$-ary $[n+1,k]$-MDS code $\mathcal{C}$ with $\dim(Hull_{E}(\mathcal{C}))=\ell$.
\end{description}
\end{theorem}

\begin{proof}
Let $a_{1}, a_{2}, \ldots, a_{n}$ be defined as \eqref{2} and $\varepsilon$ be defined as in Lemma \ref{lem3.1}. Choose $\alpha \in \mathbb{F}_{q}^{*}$ with $\alpha^{2} \neq 1$.

\textbf{(i)} Since $\frac{m}{e}$ is even, each element of $\mathbb{F}_{r}$ is a square in $\mathbb{F}_{q}$. By Lemma \ref{lem3.1}, there exist $v_{1}, \ldots,  v_{n} \in \mathbb{F}_{q}^{*}$ such that
        \[\varepsilon u_{i}= v_{i}^{2},\]
    for $1 \leq i \leq n$.  Denote $s:=k-\ell$. Put $\textbf{a}=(a_{1}, a_{2}, \ldots, a_{n})$ and $\textbf{v}=(\alpha v_{1},\ldots, \alpha v_{s}, v_{s+1},\ldots, v_{n})$. We consider the Euclidean hull of the $[n,k]_{q}$-MDS code $\mathcal{C}:=GRS_{k}(\textbf{a}, \textbf{v})$. For any $\textbf{c}=(\alpha v_{1}f(a_{1}),\ldots, \alpha v_{s}f(a_{s}),v_{s+1}f(a_{s+1}),\ldots,v_{n}f(a_{n})) \in Hull_{E}(\mathcal{C})$ with $\deg(f(x)) \leq k-1$. By Lemma \ref{lem2.1}, there exists a polynomial $g(x) \in \mathbb{F}_{q}[x]$ with $\deg(g(x)) \leq n-k-1$  such that
\begin{eqnarray*}
   & & (\alpha^{2} v_{1}^{2}f(a_{1}),\ldots, \alpha^{2} v_{s}^{2}f(a_{s}),v_{s+1}^{2}f(a_{s+1}),\ldots,v_{n}^{2}f(a_{n})) \\
    &&=(u_{1}g(a_{1}),\ldots, u_{s}g(a_{s}),u_{s+1}g(a_{s+1}),\ldots,u_{n}g(a_{n})).
\end{eqnarray*}
Since $\varepsilon u_{i}=v_{i}^{2}$, we have
\begin{equation}\label{3}
  \begin{split}
     & (\alpha^{2} \varepsilon u_{1}f(a_{1}),\ldots, \alpha^{2} \varepsilon u_{s}f(a_{s}), \varepsilon u_{s+1}f(a_{s+1}),\ldots,\varepsilon u_{n}f(a_{n})) \\
      & =(u_{1}g(a_{1}),\ldots, u_{s}g(a_{s}), u_{s+1}g(a_{s+1}),\ldots, u_{n}g(a_{n})).
  \end{split}
\end{equation}
From the last $n-s$ coordinates of Eq. \eqref{3}, we obtain that $\varepsilon f(a_{i})=g(a_{i})$ for any $s < i \leq n$. Since $k \leq \lfloor\frac{n}{2}\rfloor$, $\deg(f(x)) \leq k-1 \leq n-k-1$. Note that $\deg(g(x)) \leq n-k-1$ and $n-s \geq n-k$, thus $\varepsilon f(x)= g(x)$. On the other hand, the first $s$ coordinates of Eq. \eqref{3} imply that
\[\alpha^{2} \varepsilon u_{i}f(a_{i})= u_{i}g(a_{i})=\varepsilon u_{i}f(a_{i}),\]
for any $1 \leq i \leq s$. It follows from $\alpha^{2} \neq 1$ and $\varepsilon u_{i} \neq 0$ that $f(a_{i})=0$. Thus
\[f(x)=h(x)\prod_{i=1}^{s}(x-a_{i}),\]
for some $h(x) \in \mathbb{F}_{q}[x]$ with $\deg(h(x)) \leq k-1-s $.
It deduces that $\dim(Hull_{E}(C)) \leq k-s$.

Conversely, let $f(x)$ be a polynomial of form $h(x)\prod_{i=1}^{s}(x-a_{i})$, where $h(x) \in \mathbb{F}_{q}[x]$ and $\deg(h(x)) \leq k-1-s $. We take $g(x)=\varepsilon f(x)$, then $\deg(g(x)) \leq n-k-1$ and
\begin{eqnarray*}
\begin{split}
   & (\alpha^{2} v_{1}^{2}f(a_{1}),\ldots, \alpha^{2} v_{s}^{2}f(a_{s}),v_{s+1}^{2}f(a_{s+1}),\ldots,v_{n}^{2}f(a_{n}))\\
    & =(u_{1}g(a_{1}),\ldots, u_{s}g(a_{s}),u_{s+1}g(a_{s+1}),\ldots,u_{n}g(a_{n})).
\end{split}
\end{eqnarray*}
By Lemma \ref{lem2.1}, the vector
$(\alpha v_{1}f(a_{1}),\ldots,\alpha v_{s}f(a_{s}),$ $v_{s+1}f(a_{s+1}),\ldots,v_{n}f(a_{n})) \in Hull_{E}(\mathcal{C}).$
Therefore $\dim(Hull_{E}(\mathcal{C})) \geq k-s$, hence $\dim(Hull_{E}(\mathcal{C})) = k-s=\ell$.

\textbf{(ii)} Denote $s=k-1-\ell$. Let $\textbf{v}$ be defined as in the proof of Part (i). We consider the Euclidean hull of the $[n+1,k]_{q}$-MDS code $\mathcal{C}:=GRS_{k}(\textbf{a}, \textbf{v}, \infty)$. For any $\textbf{c}=(\alpha v_{1}f(a_{1}),\ldots, \alpha v_{s}f(a_{s}),$
$v_{s+1}f(a_{s+1}),\ldots,v_{n}f(a_{n}), f_{k-1}) \in Hull_{E}(\mathcal{C})$ with $\deg(f(x)) \leq k-1$. By Lemma \ref{lem2.2}, there exists a polynomial $g(x) \in \mathbb{F}_{q}[x]$ with $\deg(g(x)) \leq n-k$  such that
    \begin{equation}\label{4}
  \begin{split}
     & \Big(\varepsilon\alpha^{2} u_{1}f(a_{1}),\ldots, \varepsilon\alpha^{2} u_{s}f(a_{s}), \varepsilon u_{s+1}f(a_{s+1}), \ldots, \varepsilon u_{n}f(a_{n}), f_{k-1}\Big)\\
       & =\Big(u_{1}g(a_{1}),\ldots, u_{s}g(a_{s}), u_{s+1}g(a_{s+1}),\ldots, u_{n}g(a_{n}), -g_{n-k}\Big).
  \end{split}
\end{equation}
From Eq. \eqref{4}, we can similarly deduce that $ f_{k-1}= -g_{n-k}$ and $\varepsilon f(x)= g(x)$. If $f_{k-1} \neq 0$, then $k-1=n-k$, i.e., $n=2k-1$ which contradicts to the assumption that $n$ is even. Thus $f_{k-1}=0$ and $\deg(f(x)) \leq k-2$. On the other hand, the first $s$ coordinates of Eq. \eqref{4} imply that
\[\alpha^{2} \varepsilon u_{i}f(a_{i})= u_{i}g(a_{i})=\varepsilon u_{i}f(a_{i}),\]
for any $1 \leq i \leq s$. It follows from $\alpha^{2} \neq 1$ and $\varepsilon u_{i} \neq 0$ that $f(a_{i})=0$. Thus
\[f(x)=h(x)\prod_{i=1}^{s}(x-a_{i}),\]
for some $h(x) \in \mathbb{F}_{q}[x]$ with $\deg(h(x)) \leq k-2-s $.
It deduces that $\dim(Hull_{E}(C)) \leq k-1-s$.

Conversely, let $f(x)$ be a polynomial of form $h(x)\prod_{i=1}^{s}(x-a_{i})$, where $h(x) \in \mathbb{F}_{q}[x]$ and $\deg(h(x)) \leq k-2-s $. We take $g(x)=\varepsilon f(x)$, then $\deg(g(x)) \leq n-k-1$ and
\begin{eqnarray*}
\begin{split}
   & (\alpha^{2} v_{1}^{2}f(a_{1}),\ldots, \alpha^{2} v_{s}^{2}f(a_{s}),v_{s+1}^{2}f(a_{s+1}),\ldots,v_{n}^{2}f(a_{n}),0) \\
    & =(u_{1}g(a_{1}),\ldots, u_{s}g(a_{s}),u_{s+1}g(a_{s+1}),\ldots,u_{n}g(a_{n}),0).
\end{split}
\end{eqnarray*}
By Lemma \ref{lem2.2}, the vector $(\alpha v_{1}f(a_{1}),\ldots, \alpha v_{s}f(a_{s}),$ $v_{s+1}f(a_{s+1}),\ldots,v_{n}f(a_{n}), 0) \in Hull_{E}(\mathcal{C}).$
Therefore $\dim(Hull_{E}(\mathcal{C})) \geq k-1-s$, hence $\dim(Hull_{E}(\mathcal{C})) = k-1-s=\ell$.

\textbf{(iii)} We first claim that $\varepsilon$ is a square in $\mathbb{F}_{q}$. Indeed, if $q$ is even, it is done since each element in $\mathbb{F}_{q}$ is a square. Suppose $q$ is odd, if $h \neq 0 \in H$, then $-h \in H$ and $-h \neq h$. Thus $\prod_{h \in H, h \neq 0} h=(-1)^{\frac{|H|-1}{2}}\tau^{2}$ is a square in $\mathbb{F}_{q}$, where $\tau \in \mathbb{F}_{q}$. Note that $n$ is odd, thus $t$ is odd and hence $(\prod_{g \in H}(\eta-g))^{t-1}$ is a square. Thus the claim holds.   By Lemma \ref{lem3.1} and the fact that each element of $\mathbb{F}_{r}$ is a square in $\mathbb{F}_{q}$, there exist $v_{1}, \ldots,  v_{n} \in \mathbb{F}_{q}^{*}$ such that
        \[u_{i}= -v_{i}^{2},\textnormal{ for all }1 \leq i \leq n.\]
By Lemma \ref{lem3.2}, there exists a monic polynomial $\pi(x) \in \mathbb{F}_{q}[x]$ with $\deg(\pi(x))=\frac{n+1-2k}{2}$ such that
 \[\pi(a_{i}) \neq 0,\textnormal{ for all }1 \leq i \leq n.\]
 Denote $s= k-\ell$. Put $\textbf{a}=(a_{1}, a_{2}, \ldots, a_{n})$ and $\textbf{v}=(\alpha v_{1}\pi_{1},\ldots, \alpha v_{s}\pi_{s}, v_{s+1}\pi_{s+1},\ldots, v_{n}\pi_{n})$, where $\pi_{i}=\pi(a_{i})$. We consider the Euclidean hull of the $[n+1,k]_{q}$-MDS code $\mathcal{C}:=GRS_{k}(\textbf{a}, \textbf{v}, \infty)$.
For any vector $\textbf{c}=(\alpha v_{1}\pi_{1}f(a_{1}),\ldots, \alpha v_{s}\pi_{s}f(a_{s}), v_{s+1}\pi_{s+1}f(a_{s+1}),
\ldots,$ $v_{n}\pi_{n}f(a_{n}),f_{k-1}) \in Hull_{E}(\mathcal{C})$ with $\deg(f(x)) \leq k-1$. By Lemma \ref{lem2.2} and $u_{i}= -v_{i}^{2}$, there exists a polynomial $g(x) \in \mathbb{F}_{q}[x]$ with $\deg(g(x)) \leq n-k$  such that
    \begin{equation}\label{5}
  \begin{split}
     & \Big(\alpha^{2}u_{1}\pi_{1}^{2}f(a_{1}),\ldots, \alpha^{2} u_{s}\pi_{s}^{2}f(a_{s}), u_{s+1}\pi_{s+1}^{2}f(a_{s+1}), \ldots,  u_{n}\pi_{n}^{2}f(a_{n}), -f_{k-1}\Big)\\
      &=-\Big(u_{1}g(a_{1}),\ldots, u_{s}g(a_{s}), u_{s+1}g(a_{s+1}), \ldots, u_{n}g(a_{n}), -g_{n-k}\Big).
  \end{split}
\end{equation}

From the $(s+1)$-th to $n$-th coordinates of Eq. (5), we obtain that $\pi_{i}^{2}f(a_{i})=-g(a_{i})$ for any $s < i \leq n$. Since $k \leq \lfloor\frac{n}{2}\rfloor$, $\deg(\pi^{2}(x)f(x)) \leq (n+1-2k)+k-1 \leq n-k$. Note that $\deg(g(x)) \leq n-k$ and $f_{k-1}=-g_{n-k}$, thus $\deg(\pi^{2}(x)f(x)+g(x)) \leq n-k-1$. Note that $n-s \geq n-k$, hence $\pi^{2}(x)f(x)=-g(x)$. On the other hand, the first $s$ coordinates of Eq. \eqref{5} imply that
\[\alpha^{2} u_{i}\pi^{2}(a_{i})f(a_{i})= -u_{i}g(a_{i})=u_{i}\pi^{2}(a_{i})f(a_{i}),\]
for any $1 \leq i \leq s$. It follows from $\alpha^{2} \neq 1$ and $ u_{i}\pi(a_{i}) \neq 0$ that $f(a_{i})=0$. Thus
\[f(x)=h(x)\prod_{i=1}^{s}(x-a_{i}),\]
for some $h(x) \in \mathbb{F}_{q}[x]$ of $\deg(h(x)) \leq k-1-s $.
It deduces that $\dim(Hull_{E}(\mathcal{C})) \leq k-s$.

Conversely, let $f(x)$ be a polynomial of form $h(x)\prod_{i=1}^{s}(x-a_{i})$, where $h(x) \in \mathbb{F}_{q}[x]$ and $\deg(h(x)) \leq k-1-s $. We set $g(x)=-\pi^{2}(x)f(x)$. Then $\deg(g(x))\leq n-k$, and $\deg(f(x)) = k-1$ if and only if $\deg(g(x))=n-k$. Thus
$f_{k-1}=-g_{n-k}$. It is directly to verify that
\begin{eqnarray*}
\begin{split}
   & \Big(\alpha^{2} v_{1}^{2}\pi^{2}_{1}f(a_{1}),\ldots, \alpha^{2} v_{s}^{2}\pi^{2}_{s}f(a_{s}),v_{s+1}^{2}\pi^{2}_{s+1}f(a_{s+1}), \ldots,v_{n}^{2}\pi^{2}_{n}f(a_{n}), f_{k-1}\Big) \\
    &=\Big(u_{1}g(a_{1}),\ldots, u_{s}g(a_{s}), u_{s+1}g(a_{s+1}),\ldots,u_{n}g(a_{n}), -g_{n-k}\Big).
\end{split}
\end{eqnarray*}
By Lemma \ref{lem2.2}, the vector
\[(\alpha v_{1}\pi_{1}f(a_{1}),\ldots, \alpha v_{s}\pi_{s}f(a_{s}),v_{s+1}\pi_{s+1}f(a_{s+1}),\ldots,v_{n}\pi_{n}f(a_{n}), f_{k-1}) \in Hull_{E}(\mathcal{C}).\]
Therefore $\dim(Hull_{E}(\mathcal{C})) \geq k-s$, hence $\dim(Hull_{E}(\mathcal{C})) = k-s=\ell$.

The proof is completed.

\end{proof}

In the following theorem, we employ a multiplicative subgroup of $\mathbb{F}_{q}^{*}$ and the zero element to construct extended GRS codes with Euclidean hulls of arbitrary dimensions.

\begin{theorem}\label{thm3.4}
Let $q$ be a prime power. Assume that $n$ is odd, $n <q$ and $(n-1) \mid (q-1)$. If $1-n$ is a square in $\mathbb{F}_{q}$, then for any $1 \leq k \leq \frac{n+1}{2}$ and $0 \leq \ell \leq k$, there exists a $q$-ary $[n+1,k]$-MDS code $\mathcal{C}$ with $\dim(Hull_{E}(\mathcal{C}))=\ell$.
\end{theorem}

\begin{proof}
  Let $v \in \mathbb{F}_{q}^{*}$ such that $(1-n)^{-1}=v^{2}$. Let $\theta \in \mathbb{F}_{q}$ be a primitive $(n-1)$-th root of unity. For $1 \leq i \leq n-1$, denote $a_{i}=\theta^{i}$ and $a_{n}=0$. It is not hard to calculate that
  \[u_{i}=(n-1)^{-1}=-v^{2},\text{ for }1 \leq i \leq n-1,\]
  and
  \[u_{n}=-1.\]
  By Lemma \ref{lem3.2}, there exists a monic polynomial $\pi(x) \in \mathbb{F}_{q}[x]$ of $\deg(\pi(x))=\frac{n+1-2k}{2}$ such that
 \[\pi(a_{i}) \neq 0,\]
 for all $1 \leq i \leq n$. Choose $\alpha \in \mathbb{F}_{q}^{*}$ with $\alpha^{2} \neq 1$.
    Denote $s:=k-\ell$. Put $\textbf{a}=(a_{1}, a_{2}, \ldots, a_{n})$ and $\textbf{v}=(\alpha v\pi(a_{1}),\ldots, \alpha v\pi(a_{s}), v\pi(a_{s+1}),\ldots, v\pi(a_{n-1}), \pi(a_{n}))$. We consider the Euclidean hull of the $[n+1,k]_{q}$-MDS code $\mathcal{C}:=GRS_{k}(\textbf{a}, \textbf{v},\infty)$. The rest of the proof is completely similar to the Part (iii) of Theorem \ref{thm3.3}.
\end{proof}

\subsection{MDS Codes with Hermitian Hulls of Arbitrary Dimensions}
In this subsection, we will provide some constructions of $q^{2}$-ary MDS codes with Hermitian hulls of arbitrary dimensions.

The first construction consider the $q^{2}$-ary MDS codes of length $n \leq q$.
\begin{theorem}\label{thm3.5}
Let $q>2$ be a prime power. Assume that $n \leq q$. Then for any $1 \leq k \leq \lfloor\frac{n}{2}\rfloor$ and $0 \leq \ell \leq k$, there exists a $q^{2}$-ary $[n,k]$-MDS code $\mathcal{C}$ with $\dim(Hull_{H}(\mathcal{C}))=\ell$.
\end{theorem}

\begin{proof}
   Let $a_{1}, a_{2}, \ldots, a_{n}$ be $n$ distinct elements of $\mathbb{F}_{q}$ and $u_{i}$ be defined as in Eq. \eqref{1}. Thus for $1 \leq i \leq n$, we have $u_{i} \in \mathbb{F}_{q}^{*}$ and hence there exist $v_{1}, \ldots,  v_{n} \in \mathbb{F}_{q^{2}}^{*}$ such that
  \[u_{i}=v_{i}^{q+1}.\]
  Choose $\alpha \in \mathbb{F}_{q^{2}}^{*}$ such that $\beta:=\alpha ^{q+1} \neq 1$. Denote $s:=k-\ell$. Put $\textbf{a}=(a_{1}, a_{2}, \ldots, a_{n})$ and $\textbf{v}=(\alpha v_{1},\ldots, \alpha v_{s}, v_{s+1},\ldots, v_{n})$. We consider the Hermitian hull of the $[n,k]_{q^{2}}$-MDS code $\mathcal{C}:=GRS_{k}(\textbf{a}, \textbf{v})$. For any $\textbf{c}=(\alpha v_{1}f(a_{1}),$ $\ldots, \alpha v_{s}f(a_{s}),v_{s+1}f(a_{s+1}),\ldots,v_{n}f(a_{n})) \in Hull_{H}(\mathcal{C})$ with $\deg(f(x)) \leq k-1$. By Lemma \ref{lem2.3}, there exists a polynomial $g(x) \in \mathbb{F}_{q^{2}}[x]$ with $\deg(g(x)) \leq n-k-1$  such that
\begin{eqnarray*}
\begin{split}
   & \Big(\alpha^{q+1} v_{1}^{q+1}f^{q}(a_{1}),\ldots, \alpha^{q+1} v_{s}^{q+1}f^{q}(a_{s}),v_{s+1}^{q+1}f^{q}(a_{s+1}), \ldots,v_{n}^{q+1}f^{q}(a_{n})\Big) \\
    & =\Big(u_{1}g(a_{1}),\ldots, u_{s}g(a_{s}),u_{s+1}g(a_{s+1}), \ldots,u_{n}g(a_{n})\Big),
\end{split}
\end{eqnarray*}
i.e.,
\begin{equation}\label{6}
  \begin{split}
     & (\beta u_{1}f^{q}(a_{1}),\ldots, \beta u_{s}f^{q}(a_{s}),u_{s+1}f^{q}(a_{s+1}),\ldots, u_{n}f^{q}(a_{n})) \\
      & =(u_{1}g(a_{1}),\ldots, u_{s}g(a_{s}), u_{s+1}g(a_{s+1}),\ldots, u_{n}g(a_{n})).
  \end{split}
\end{equation}
Write $f(x)=\sum_{i=0}^{k-1}f_{i}x^{i}$. Denote $F(x)=\sum_{i=0}^{k-1}f_{i}^{q}x^{i}$. Since $a_{i} \in \mathbb{F}_{q}$, $F(a_{i})=f^{q}(a_{i})$, for $1 \leq i\leq n$. From Eq. \eqref{6}, we have
\begin{equation}\label{7}
  \begin{split}
     & (\beta u_{1}F(a_{1}),\ldots, \beta u_{s}F(a_{s}),u_{s+1}F(a_{s+1}),\ldots, u_{n}F(a_{n})) \\
      & =(u_{1}g(a_{1}),\ldots, u_{s}g(a_{s}), u_{s+1}g(a_{s+1}),\ldots, u_{n}g(a_{n})).
  \end{split}
\end{equation}
From the last $n-s$ coordinates of Eq. \eqref{7}, we obtain $F(a_{i})=g(a_{i})$ for any $s < i \leq n$. Since $k \leq \lfloor\frac{n}{2}\rfloor$, $\deg(F(x)) \leq k-1 \leq n-k-1$. Note that $\deg(g(x)) \leq n-k-1$ and $n-s \geq n-k$, thus $F(x)= g(x)$. On the other hand, the first $s$ coordinates of Eq. \eqref{7} imply that
\[\beta u_{i}F(a_{i})= u_{i}g(a_{i})= u_{i}F(a_{i}),\]
for any $1 \leq i \leq s$. It follows from $\beta \neq 1$ and $u_{i} \neq 0$ that $F(a_{i})=0$, hence $f(a_{i})=0$. Thus
\[f(x)=h(x)\prod_{i=1}^{s}(x-a_{i}),\]
for some $h(x) \in \mathbb{F}_{q^{2}}[x]$ of $\deg(h(x)) \leq k-1-s $.
It deduces that $\dim(Hull_{H}(\mathcal{C})) \leq k-s$.

Conversely, let $f(x)$ be a polynomial of form $h(x)\prod_{i=1}^{s}(x-a_{i})$, where $h(x) \in \mathbb{F}_{q^{2}}[x]$ and $\deg(h(x)) \leq k-1-s $. We take $g(x)=F(x)$, then $\deg(g(x)) \leq k-1 \leq n-k-1$ and $g(a_{i})=F(a_{i})=f^{q}(a_{i})$. Thus
\begin{eqnarray*}
\begin{split}
  & \Big(\alpha^{q+1} v_{1}^{q+1}f^{q}(a_{1}),\ldots, \alpha^{q+1} v_{s}^{q+1}f^{q}(a_{s}),v_{s+1}^{q+1}f^{q}(a_{s+1}), \ldots,v_{n}^{q+1}f^{q}(a_{n})\Big)\\
    &=\Big(u_{1}g(a_{1}),\ldots, u_{s}g(a_{s}),u_{s+1}g(a_{s+1}), \ldots,u_{n}g(a_{n})\Big).
\end{split}
\end{eqnarray*}
By Lemma \ref{lem2.3}, the vector $(\alpha v_{1}f(a_{1}),\ldots, \alpha v_{s}f(a_{s}),v_{s+1}f(a_{s+1}),\ldots,v_{n}f(a_{n})) \in Hull_{H}(\mathcal{\mathcal{C}}).$
Therefore $\dim(Hull_{H}(\mathcal{C})) \geq k-s$, hence $\dim(Hull_{H}(\mathcal{C})) = k-s=\ell$.

\end{proof}

Similarly as the Euclidean case in Theorem \ref{thm3.3}, we can use an $\mathbb{F}_{r}$-subspace of $\mathbb{F}_{q^{2}}$ and its cosets to construct GRS codes with Hermitian hulls of arbitrary dimensions as follows.
\begin{theorem}\label{thm3.6}
 Let $q=p^{m} \geq 3$ and $r=p^{e}$, where $e \mid m$. Let $n=tr^{z}$, where $1 \leq t \leq r$ and $1 \leq z \leq 2\frac{m}{e}-1$. Then
 \begin{description}
   \item[(i)] for any $1 \leq k \leq \lfloor\frac{n-1+q}{q+1}\rfloor$ and $0 \leq \ell \leq k$, there exists a $q^{2}$-ary $[n,k]$-MDS code $\mathcal{C}$ with $\dim(Hull_{H}(\mathcal{C}))=\ell$;
   \item[(ii)] for any $1 \leq k \leq \lfloor\frac{n-1+q}{q+1}\rfloor$ and $0 \leq \ell \leq k-1$, there exists a $q^{2}$-ary $[n+1,k]$-MDS code $\mathcal{C}$ with $\dim(Hull_{H}(\mathcal{C}))=\ell$.
 \end{description}

\end{theorem}
\begin{proof}
  Let $H$ be an $\mathbb{F}_{r}$-subspace of $\mathbb{F}_{q^{2}}$ of dimension $z$ and $\eta \in \mathbb{F}_{q^{2}} \setminus H$. We can define the subset $\{a_{1}, a_{2}, \ldots, a_{n}\}$ of $\mathbb{F}_{q^{2}}$ similarly as Eq. \eqref{2} of Subsection 3.1.  Choose $\alpha \in \mathbb{F}_{q^{2}}^{*}$ such that $\beta:=\alpha^{q+1} \neq 1$. Let $\varepsilon$ be defined as in Lemma \ref{lem3.1}. Note that each element of $\mathbb{F}_{q}$ is a $(q+1)$-th power in $\mathbb{F}_{q^{2}}$.  Thus, by Lemma \ref{lem3.1},
  there exist $v_{1}, \ldots,  v_{n} \in \mathbb{F}_{q^{2}}^{*}$ such that
        \[ v_{i}^{q+1}=\varepsilon u_{i}, \textnormal{ for } 1 \leq i \leq n.\]

    \textbf{(i)} Denote $s:=k-\ell$. Put $\textbf{a}=(a_{1}, a_{2}, \ldots, a_{n})$ and $\textbf{v}=(\alpha v_{1},\ldots, \alpha v_{s}, v_{s+1},\ldots, v_{n})$. We consider the Hermitian hull of the $[n,k]_{q^{2}}$-MDS code $\mathcal{C}:=GRS_{k}(\textbf{a}, \textbf{v})$. For any $\textbf{c}=(\alpha v_{1}f(a_{1}),$ $\ldots, \alpha v_{s}f(a_{s}),v_{s+1}f(a_{s+1}),\ldots,v_{n}f(a_{n})) \in Hull_{H}(\mathcal{C})$, where $\deg(f(x)) \leq k-1$. By Lemma \ref{lem2.3}, there exists a polynomial $g(x) \in \mathbb{F}_{q^{2}}[x]$ with $\deg(g(x)) \leq n-k-1$  such that
\begin{eqnarray*}
\begin{split}
   & \Big(\alpha^{q+1} v_{1}^{q+1}f^{q}(a_{1}),\ldots, \alpha^{q+1} v_{s}^{q+1}f^{q}(a_{s}),v_{s+1}^{q+1}f^{q}(a_{s+1}), \ldots,v_{n}^{q+1}f^{q}(a_{n})\Big)\\
    & =\Big(u_{1}g(a_{1}),\ldots, u_{s}g(a_{s}),u_{s+1}g(a_{s+1}), \ldots,u_{n}g(a_{n})\Big),
\end{split}
\end{eqnarray*}
i.e.,
\begin{equation}\label{8}
  \begin{split}
     & \varepsilon(\beta u_{1}f^{q}(a_{1}),\ldots, \beta u_{s}f^{q}(a_{s}),u_{s+1}f^{q}(a_{s+1}),\ldots, u_{n}f^{q}(a_{n})) \\
      & =(u_{1}g(a_{1}),\ldots, u_{s}g(a_{s}), u_{s+1}g(a_{s+1}),\ldots, u_{n}g(a_{n})).
  \end{split}
\end{equation}
From the last $n-s$ coordinates of Eq. \eqref{8}, we obtain that $\varepsilon f^{q}(a_{i})=g(a_{i})$ for any $s < i \leq n$. Since $k \leq \lfloor\frac{n-1+q}{q+1}\rfloor$, $\deg(f^{q}(x)) \leq q(k-1) \leq n-k-1$. Note that $\deg(g(x)) \leq n-k-1$ and $n-s \geq n-k$, thus $\varepsilon f^{q}(x)= g(x)$. On the other hand, the first $s$ coordinates of Eq. \eqref{8} imply that
\[\varepsilon\beta u_{i}f^{q}(a_{i})= u_{i}g(a_{i})= u_{i}\varepsilon f^{q}(a_{i}),\]
for any $1 \leq i \leq s$. It follows from $\beta \neq 1$ and $\varepsilon u_{i} \neq 0$ that $f^{q}(a_{i})=0$, i.e., $f(a_{i})=0$. Thus
\[f(x)=h(x)\prod_{i=1}^{s}(x-a_{i}),\]
for some $h(x) \in \mathbb{F}_{q^{2}}[x]$ of $\deg(h(x)) \leq k-1-s $.
It deduces that $\dim(Hull_{H}(\mathcal{C})) \leq k-s$.

Conversely, let $f(x)$ be a polynomial of form $h(x)\prod_{i=1}^{s}(x-a_{i})$, where $h(x) \in \mathbb{F}_{q^{2}}[x]$ and $\deg(h(x)) \leq k-1-s $. We take $g(x)=\varepsilon f^{q}(x)$, then $\deg(g(x)) \leq q(k-1) \leq n-k-1$ and
\begin{eqnarray*}
\begin{split}
  & \Big(\alpha^{q+1} v_{1}^{q+1}f^{q}(a_{1}),\ldots, \alpha^{q+1} v_{s}^{q+1}f^{q}(a_{s}),v_{s+1}^{q+1}f^{q}(a_{s+1}), \ldots,v_{n}^{q+1}f^{q}(a_{n})\Big)\\
    &=\Big(u_{1}g(a_{1}),\ldots, u_{s}g(a_{s}),u_{s+1}g(a_{s+1}), \ldots,u_{n}g(a_{n})\Big).
\end{split}
\end{eqnarray*}
By Lemma \ref{lem2.3}, the vector $(\alpha v_{1}f(a_{1}),\ldots, \alpha v_{s}f(a_{s}),v_{s+1}f(a_{s+1}),\ldots,v_{n}f(a_{n})) \in Hull_{H}(\mathcal{C}).$
Therefore $\dim(Hull_{H}(\mathcal{C})) \geq k-s$, hence $\dim(Hull_{H}(\mathcal{C})) = k-s=\ell$.

\textbf{(ii)} Denote $s:=k-1-\ell$. Put $\textbf{a}=(a_{1}, a_{2}, \ldots, a_{n})$ and $\textbf{v}=(\alpha v_{1},\ldots, \alpha v_{s}, v_{s+1},\ldots, v_{n})$. We consider the Hermitian hull of the $[n+1,k]_{q^{2}}$-MDS code $\mathcal{C}:=GRS_{k}(\textbf{a}, \textbf{v}, \infty)$. For any $\textbf{c}=(\alpha v_{1}f(a_{1}),\ldots, \alpha v_{s}f(a_{s}),v_{s+1}f(a_{s+1}),\ldots,v_{n}f(a_{n}), f_{k-1}) \in Hull_{H}(\mathcal{C})$, where $\deg(f(x)) \leq k-1$. By Lemma \ref{lem2.4}, there exists a polynomial $g(x) \in \mathbb{F}_{q^{2}}[x]$ with $\deg(g(x)) \leq n-k$  such that
\begin{equation}\label{9}
  \begin{split}
     & \Big(\varepsilon\beta u_{1}f^{q}(a_{1}),\ldots, \varepsilon\beta u_{s}f^{q}(a_{s}),\varepsilon u_{s+1}f^{q}(a_{s+1}), \ldots, \varepsilon u_{n}f^{q}(a_{n}), f^{q}_{k-1}\Big)\\
      &  =\Big(u_{1}g(a_{1}),\ldots, u_{s}g(a_{s}), u_{s+1}g(a_{s+1}), \ldots, u_{n}g(a_{n}),-g_{n-k}\Big).
  \end{split}
\end{equation}
Then from Eq. \eqref{9}, we can similarly deduce that $f^{q}_{k-1}=-g_{n-k}$ and  $\varepsilon f^{q}(x)= g(x)$. If $f_{k-1} \neq 0$, then $q(k-1)=n-k$, which contradicts to the assumption that $k \leq \lfloor\frac{n-1+q}{q+1}\rfloor$. Thus $f_{k-1}=0$, i.e., $\deg(f(x)) \leq k-2$. On the other hand, the first $s$ coordinates of Eq. \eqref{9} imply that
\[\varepsilon\beta u_{i}f^{q}(a_{i})= u_{i}g(a_{i})=\varepsilon u_{i}f^{q}(a_{i}),\]
for any $1 \leq i \leq s$. It follows from $\beta \neq 1$ and $\varepsilon u_{i} \neq 0$ that $f^{q}(a_{i})=0$, i.e., $f(a_{i})=0$. Thus
\[f(x)=h(x)\prod_{i=1}^{s}(x-a_{i}),\]
for some $h(x) \in \mathbb{F}_{q^{2}}[x]$ of $\deg(h(x)) \leq k-2-s $.
It deduces that $\dim(Hull_{H}(\mathcal{C})) \leq k-1-s$.

Conversely, let $f(x)$ be a polynomial of form $h(x)\prod_{i=1}^{s}(x-a_{i})$, where $h(x) \in \mathbb{F}_{q^{2}}[x]$ and $\deg(h(x)) \leq k-2-s $. We can similarly prove that
\[(\alpha v_{1}f(a_{1}),\ldots, \alpha v_{s}f(a_{s}),v_{s+1}f(a_{s+1}),\ldots,v_{n}f(a_{n}),0) \in Hull_{H}(\mathcal{C}).\]
Thus $\dim(Hull_{H}(\mathcal{C})) \geq k-1-s$, hence $\dim(Hull_{H}(\mathcal{C})) = k-1-s=\ell$.

The proof is completed.
\end{proof}

In the following, we consider a multiplicative subgroup of $\mathbb{F}_{q^{2}}^{*}$ and its cosets. Suppose $n' \mid (q^{2}-1)$. We write $n'=\frac{n'}{\gcd(n', q+1)}\cdot \gcd(n', q+1)$. For convenience, we denote $n_{1}=\frac{n'}{\gcd(n',q+1)}$ and $n_{2}=\frac{n'}{n_{1}}=\gcd(n', q+1)$. Then $n_{1} \mid (q-1)\frac{q+1}{n_{2}}$. Note that $\gcd(n_{1}, \frac{q+1}{n_{2}})=\gcd(\frac{n'}{n_{2}}, \frac{q+1}{n_{2}})=1$, hence $n_{1} \mid (q-1)$. Let $\omega$ be a primitive element of $\mathbb{F}_{q^{2}}$. Let $G$ and $H$ be the subgroups of $\mathbb{F}_{q^{2}}^{*}$ generated by $\omega^{\frac{q^{2}-1}{n'}}$ and $\omega^{\frac{q+1}{n_{2}}}$, respectively. Then $|G|=n'$ and $|H|=(q-1)n_{2}$. Note that $\frac{q^{2}-1}{n'}=\frac{q+1}{n_{2}}\cdot\frac{q-1}{n_{1}}$, thus $ \frac{q+1}{n_{2}} \mid \frac{q^{2}-1}{n'}$, which deduce that $G$ is a subgroup of $H$.
Then there exist $\beta_{1}, \ldots , \beta_{\frac{q-1}{n_{1}}} \in H$
such that $\{\beta_{b} G\}^{\frac{q-1}{n_{1}}}_{b
=1}$ represent all cosets of $H/G$.

Now, let $n=tn'$, where $1 \leq t \leq \frac{q-1}{n_{1}}$. Set $A_{b}=\beta_{b}G$ and $A=\bigcup^{t}_{b=1}A_{b}$. Suppose

 \begin{equation}\label{10}
   A:=\{a_{1}, a_{2}, \ldots , a_{n}\}.
\end{equation}
 We calculate the value of $u_{i}$ defined by Eq. \eqref{1} as follows.

\begin{lemma}\label{lem3.7}
  Keep the notations as above. Given $1 \leq i \leq n$, suppose $a_{i} \in A_{b}$ for some $1 \leq b \leq t$. Then
  \[u_{i}=\frac{1}{n'}a_{i}\beta_{b}^{-n'}\prod_{1 \leq s \leq t, s \neq b}(\beta_{b}^{n'}-\beta_{s}^{n'})^{-1}.\]
  Moreover, we have $a_{i}^{-1}u_{i} \in \mathbb{F}^{*}_{q}$.
\end{lemma}

\begin{proof}
  Let $\theta=\omega^{\frac{q^{2}-1}{n'}}$ be the generator of $G$. Suppose $a_{i}=\beta_{b}\theta^{\ell}$ for some $1 \leq \ell \leq n'$. Then
  \[u_{i}=\prod_{x_{b} \in A_{b}}(a_{i}-x_{b})^{-1}\prod_{1 \leq s \leq t, s \neq b}\prod_{x_{s} \in A_{s}}(a_{i}-x_{s})^{-1}.\]
  Since $\prod\limits_{1 \leq j \leq n'-1}(1-\theta^{j})=n'$, we have
  \begin{eqnarray*}
    \prod_{x_{b} \in A_{b}, x_{b} \neq a_{i}}(a_{i}-x_{b}) &=& \prod_{1 \leq j \leq n', j \neq \ell}(\beta_{b}\theta^{\ell}-\beta_{b}\theta^{j}) \\
      &=& (\beta_{b}\theta^{\ell})^{n'-1}\prod_{1 \leq j \leq n'-1}(1-\theta^{j})\\
      &=&n'a_{i}^{-1}\beta_{b}^{n'}.
  \end{eqnarray*}
Since $\prod\limits_{1 \leq j \leq n'}(x-\beta\theta^{j})=x^{n'}-\beta^{n'}$, we have
\[\prod_{x_{s} \in A_{s}}(a_{i}-x_{s})=\prod_{1 \leq j \leq n'}(\beta_{b}\theta^{\ell}-\beta_{s}\theta^{j})=\beta_{b}^{n'}-\beta_{s}^{n'}. \]
The first conclusion then follows. For any $\beta_{s} \in H$, there exists an integer $j$ such that $\beta_{s}=\omega^{j\frac{q+1}{n_{2}}}$. Thus
$\beta_{s}^{n'}=\omega^{jn'\frac{q+1}{n_{2}}}=\omega^{jn_{1}(q+1)}$, which is an element of $\mathbb{F}^{*}_{q}$. The second conclusion then follows from the first conclusion.
\end{proof}
From the above discussions, we provide a construction of GRS codes with Hermitian hulls of arbitrary dimensions as follows.
\begin{theorem}\label{thm3.8}
Let $q>2$ be a prime power and $n' \mid (q^{2}-1)$. Let $n=tn'$, where $1 \leq t \leq \frac{q-1}{n_{1}}$ and $n_{1}=\frac{n'}{\gcd(n',q+1)}$. Then
 for any $1 \leq k \leq \lfloor\frac{n+q}{q+1}\rfloor$ and $0 \leq \ell \leq k-1$, there exists a $q^{2}$-ary $[n,k]$-MDS code $\mathcal{C}$ with $\dim(Hull_{H}(\mathcal{C}))=\ell$.
\end{theorem}

\begin{proof}
Let $a_{1}, a_{2}, \ldots , a_{n}$ be defined as Eq. \eqref{10}. From Lemma \ref{lem3.7}, there exist $v_{1}, \ldots,  v_{n} \in \mathbb{F}_{q^{2}}^{*}$ such that
        \[ a_{i}^{-1}u_{i}= v_{i}^{q+1},\textnormal{ for all } 1 \leq i \leq n.\]
         Let $\alpha \in \mathbb{F}_{q^{2}}^{*}$ such that $\beta:=\alpha^{q+1} \neq 1$. Denote $s:=k-1-\ell$. Put $\textbf{a}=(a_{1}, a_{2}, \ldots, a_{n})$ and $\textbf{v}=(\alpha v_{1},\ldots, \alpha v_{s}, v_{s+1},\ldots, v_{n})$. We consider the Hermitian hull of the $[n,k]_{q^{2}}$-MDS code $\mathcal{C}:=GRS_{k}(\textbf{a}, \textbf{v})$. For any $\textbf{c}=(\alpha v_{1}f(a_{1}),\ldots, \alpha v_{s}f(a_{s}),v_{s+1}f(a_{s+1}),\ldots, v_{n}f(a_{n})) \in Hull_{H}(\mathcal{C})$ with $\deg(f(x)) \leq k-1$. By Lemma \ref{lem2.3}, there exists a polynomial $g(x) \in \mathbb{F}_{q^{2}}[x]$ with $\deg(g(x)) \leq n-k-1$  such that
\begin{eqnarray*}
\begin{split}
   & \Big(\alpha^{q+1} v_{1}^{q+1}f^{q}(a_{1}),\ldots, \alpha^{q+1} v_{s}^{q+1}f^{q}(a_{s}),v_{s+1}^{q+1}f^{q}(a_{s+1}), \ldots, v_{n}^{q+1}f^{q}(a_{n})\Big) \\
    &=\Big(u_{1}g(a_{1}),\ldots, u_{s}g(a_{s}),u_{s+1}g(a_{s+1}), \ldots,u_{n}g(a_{n})\Big),
\end{split}
\end{eqnarray*}

i.e.,
\begin{equation}\label{11}
  \begin{split}
     & \Big(\beta u_{1}f^{q}(a_{1}),\ldots, \beta u_{s}f^{q}(a_{s}),u_{s+1}f^{q}(a_{s+1}), \ldots, u_{n}f^{q}(a_{n})\Big) \\
      & =\Big(a_{1}u_{1}g(a_{1}),\ldots, a_{s}u_{s}g(a_{s}), a_{s+1}u_{s+1}g(a_{s+1}),\ldots, a_{n}u_{n}g(a_{n})\Big).
  \end{split}
\end{equation}
From the last $n-s$ coordinates of Eq. \eqref{11}, we obtain that $f^{q}(a_{i})=a_{i}g(a_{i})$ for any $s < i \leq n$. Since $k \leq \lfloor\frac{n+q}{q+1}\rfloor$, $\deg(f^{q}(x)) \leq q(k-1) \leq n-k$. Note that $\deg(xg(x)) \leq n-k$ and $n-s=n-k+\ell+1 \geq n-k+1$, thus $f^{q}(x)=xg(x)$. Hence, $x \mid f(x)$. On the other hand, the first $s$ coordinates of Eq. \eqref{11} imply that
\[\beta u_{i}f^{q}(a_{i})= u_{i}a_{i}g(a_{i})= u_{i}f^{q}(a_{i}),\]
for any $1 \leq i \leq s$. It follows from $\beta \neq 1$ and $u_{i} \neq 0$ that $f^{q}(a_{i})=0$, i.e., $f(a_{i})=0$. Thus
\[f(x)=xh(x)\prod_{i=1}^{s}(x-a_{i}),\]
for some $h(x) \in \mathbb{F}_{q^{2}}[x]$ of $\deg(h(x)) \leq k-2-s $.
It deduces that $\dim(Hull_{H}(\mathcal{C})) \leq k-1-s$.

Conversely, let $f(x)$ be a polynomial of form $xh(x)\prod_{i=1}^{s}(x-a_{i})$, where $h(x) \in \mathbb{F}_{q^{2}}[x]$ and $\deg(h(x)) \leq k-2-s $. We take $g(x)=\frac{f^{q}(x)}{x}$, then $\deg(g(x)) \leq q(k-1)-1 \leq n-k-1$ and

\begin{eqnarray*}
\begin{split}
  & \Big(\alpha^{q+1} v_{1}^{q+1}f^{q}(a_{1}),\ldots, \alpha^{q+1} v_{s}^{q+1}f^{q}(a_{s}),v_{s+1}^{q+1}f^{q}(a_{s+1}), \ldots, v_{n}^{q+1}f^{q}(a_{n})\Big)  \\
    & =\Big(u_{1}g(a_{1}),\ldots, u_{s}g(a_{s}),u_{s+1}g(a_{s+1}), \ldots,u_{n}g(a_{n})\Big).
\end{split}
\end{eqnarray*}
By Lemma \ref{lem2.3},
\[(v_{1}f(a_{1}),\ldots, v_{s}f(a_{s}), f(a_{s+1}),\ldots, f(a_{n})) \in Hull_{H}(\mathcal{C}).\]
Therefore $\dim(Hull_{H}(\mathcal{C})) \geq k-1-s$, hence $\dim(Hull_{H}(\mathcal{C})) = k-1-s=\ell$. The proof is completed.
\end{proof}

In Theorem \ref{thm3.8}, by adding the zero element, we can obtain a family of GRS codes of length $n+1$ with Hermitian hulls of arbitrary dimensions.
\begin{theorem}\label{thm3.9}
 Let $q>2$ be a prime power and $n' \mid (q^{2}-1)$. Let $n=tn'$, where $1 \leq t \leq \frac{q-1}{n_{1}}$ and $n_{1}=\frac{n'}{\gcd(n',q+1)}$. Then
 for any $1 \leq k \leq \lfloor\frac{n+q}{q+1}\rfloor$ and $0 \leq \ell \leq k$, there exists a $q^{2}$-ary $[n+1,k]$-MDS code $\mathcal{C}$ with $\dim(Hull_{H}(\mathcal{C}))=\ell$.
\end{theorem}
\begin{proof}
  Let $a_{1}, a_{2}, \ldots , a_{n}$ be defined as in Theorem \ref{thm3.8}. Put $a_{n+1}=0$. From Lemma \ref{lem3.7}, for any $1 \leq i \leq n$ it is easy to see that
  \[\prod_{1 \leq j \leq n+1, j \neq i}(a_{i}-a_{j})^{-1}=a_{i}^{-1}\prod_{1 \leq j \leq n, j \neq i}(a_{i}-a_{j})^{-1} \in \mathbb{F}_{q}^{*}.\]
  And note that
  \[\prod_{j=1}^{n}(a_{n+1}-a_{j})^{-1}=(-1)^{n}\prod_{j=1}^{n}a_{j}^{-1}=(-1)^{n+1+n't}\prod_{j=1}^{t}\beta_{j}^{-n'}\]
  is also an element of $\mathbb{F}_{q}^{*}$ since $\beta_{j}^{-n'} \in \mathbb{F}_{q}^{*}$.
We still denote $\prod\limits_{1 \leq j \leq n+1, j \neq i}(a_{i}-a_{j})^{-1}$ by $u_{i}$.
  Then, for $1 \leq i \leq n+1$, there exists $v_{i} \in \mathbb{F}_{q^{2}}^{*}$ such that
  \[u_{i}=v_{i}^{q+1}.\]
Let $\alpha \in \mathbb{F}_{q^{2}}^{*}$ such that $\beta:=\alpha^{q+1} \neq 1$. Denote $s:=k-\ell$. Put $\textbf{a}=(a_{1}, a_{2}, \ldots, a_{n+1})$ and $\textbf{v}=(\alpha v_{1},\ldots, \alpha v_{s}, v_{s+1},\ldots, v_{n+1})$. We consider the Hermitian hull of the $[n+1,k]_{q^{2}}$-MDS code $\mathcal{C}:=GRS_{k}(\textbf{a}, \textbf{v})$. The rest of the proof is completely similar to the Part (i) of Theorem \ref{thm3.6}.
\end{proof}

We extend the GRS codes in Theorem \ref{thm3.9} to obtain a family of extended GRS codes of length $n+2$ with Hermitian hulls of arbitrary dimensions in the following theorem.
\begin{theorem}\label{thm3.10}
 Let $q>2$ be a prime power and $n' \mid (q^{2}-1)$. Let $n=tn'$, where $1 \leq t \leq \frac{q-1}{n_{1}}$ and $n_{1}=\frac{n'}{\gcd(n',q+1)}$. Then
 for any $1 \leq k \leq \lfloor\frac{n+q}{q+1}\rfloor$ and $0 \leq \ell \leq k-1$, there exists a $q^{2}$-ary $[n+2,k]$-MDS code $\mathcal{C}$ with $\dim(Hull_{H}(\mathcal{C}))=\ell$.
\end{theorem}

\begin{proof}
  Let $\textbf{a}=(a_{1}, a_{2}, \ldots, a_{n+1})$ and $\textbf{v}=(\alpha v_{1},\ldots, \alpha v_{s}, v_{s+1},\ldots, v_{n+1})$ be defined as in the proof of Theorem \ref{thm3.9}. We consider the $[n+2,k]_{q^{2}}$-MDS code $\mathcal{C}:=GRS_{k}(\textbf{a},\textbf{v},\infty)$. With the same argument of the proof of Part (ii) of Theorem \ref{thm3.6}, we can show that $\dim(Hull_{H}(\mathcal{C}))=\ell$.
\end{proof}

By the classical MDS conjecture, the length of an MDS code over $\mathbb{F}_{q^{2}}$ is bounded by $q^{2}+1$ (except for two specific cases). Taking $(t, r, b)=(q,q,1)$ and $(t, n')=(q-1, q+1)$ in Theorems \ref{thm3.6} and \ref{thm3.10}, respectively, MDS codes of length $q^{2}+1$ and dimension $1 \leq k \leq q-1$ with Hermitian hull of dimension $0 \leq \ell \leq k-1$ are obtained. In the following, we consider the MDS code of length $q^{2}+1$ and dimension $q$.

\begin{theorem}\label{thm3.11}
 Let $q>2$ be a prime power. Then for any  $0 \leq \ell \leq q$, there exists a $q^{2}$-ary $[q^{2}+1,q]$-MDS code $\mathcal{C}$ with $\dim(Hull_{H}(\mathcal{C}))=\ell$.
\end{theorem}

\begin{proof}
  Suppose $\mathbb{F}_{q^{2}}=\{a_{1},a_{2}, \cdots, a_{q^{2}} \}$. It is easy to show that
  \[u_{i}:=\prod_{1 \leq j \leq q^{2}, j \neq i}(a_{i}-a_{j})^{-1}=-1, \textnormal{ for all } 1 \leq i \leq q^{2}.\]
  Denote $s:=q-\ell$. For $1 \leq i \leq s$, let $v_{i} \in \mathbb{F}_{q^{2}}$ such that $v_{i}^{q+1} \neq 1$. Put $\textbf{a}=(a_{1}, a_{2}, \ldots, a_{q^{2}})$ and $\textbf{v}=(v_{1},\ldots, v_{s}, 1,\ldots, 1)$. We consider the Hermitian hull of the $[q^{2}+1, q]_{q^{2}}$-MDS code $\mathcal{C}:=GRS_{q}(\textbf{a},\textbf{v}, \infty)$. For any $\textbf{c}=(v_{1}f(a_{1}),\ldots, v_{s}f(a_{s}),f(a_{s+1}),\ldots, f(a_{q^{2}}), f_{q-1}) \in Hull_{H}(\mathcal{C})$ with $\deg(f(x)) \leq q-1$. By Lemma \ref{lem2.4}, there exists a polynomial $g(x) \in \mathbb{F}_{q^{2}}[x]$ with $\deg(g(x)) \leq q^{2}-q$  such that

\begin{equation}\label{12}
  \begin{split}
     & (v_{1}^{q+1}f^{q}(a_{1}),\ldots, v_{s}^{q+1}f^{q}(a_{s}),f^{q}(a_{s+1}),\ldots, f^{q}(a_{q^{2}}), f^{q}_{q-1}) \\
    &=-(g(a_{1}),\ldots, g(a_{s}),g(a_{s+1}),\ldots,g(a_{q^{2}}),g_{q^{2}-q}).
  \end{split}
\end{equation}

Thus we have $f^{q}_{q-1}=-g_{q^{2}-q}$ and $f^{q}(a_{i})=-g(a_{i})$ for any $s < i \leq q^{2}$. Note that $\deg(f^{q}(x)) \leq q(q-1)$. Thus $\deg(f^{q}(x)+g(x)) \leq q^{2}-q-1$ since $f^{q}_{q-1}=-g_{q^{2}-q}$. Note that $q^{2}-s \geq q^{2}-q$, we have $f^{q}(x)=-g(x)$. On the other hand, the first $s$ coordinates of Eq. \eqref{12} imply that
\[v_{i}^{q+1}f^{q}(a_{i})= -g(a_{i})= f^{q}(a_{i}),\]
for any $1 \leq i \leq s$. It follows from $v_{i}^{q+1} \neq 1$  that $f^{q}(a_{i})=0$, i.e., $f(a_{i})=0$. Thus
\[f(x)=h(x)\prod_{i=1}^{s}(x-a_{i}),\]
for some $h(x) \in \mathbb{F}_{q^{2}}[x]$ of $\deg(h(x)) \leq q-1-s $.
It deduces that $\dim(Hull_{H}(\mathcal{C})) \leq q-s$.

Conversely, we can similarly show that $\dim(Hull_{H}(\mathcal{C})) \geq q-s$, hence $\dim(Hull_{H}(\mathcal{C})) = q-s= \ell$. The proof is completed.
\end{proof}
\section{Applications to EAQECCs}

In this section, we introduce some basic notions of entanglement-assisted quantum error-correcting codes (EAQECCs) and then construct several new families of MDS EAQECCs by employing the results in Section 3. For more details on EAQECCs, we
refer the reader to \cite{WB08,LA18,HDB07,SHB11,LB13}.

First, we recall some basics of quantum codes. In a quantum system,  a quantum state is called a qubit. Let $\mathbb{C}$ be the complex field and $\mathbb{C}^{q}$ be the $q$-dimensional Hilbert space over $\mathbb{C}$. A qubit is just a non-zero vector of $\mathbb{C}^{q}$. Let $\{|a\rangle : a \in \mathbb{F}_{q}\}$ be a basis of $\mathbb{C}^{q}$, then a qubit $|v\rangle$ can be expressed as
\[|v\rangle=\sum_{a \in \mathbb{F}_{q}}v_{a}|a\rangle,\]
where $v_{a} \in \mathbb{C}.$ In general, an $n$-qubit is a joint state of $n$ qubits in the $q^{n}$-dimensional Hilbert space $(\mathbb{C}^{q})^{\bigotimes n}\cong \mathbb{C}^{q^{n}}$. Similarly, an $n$-qubit can be represented as
\[|\textbf{v}\rangle=\sum_{\textbf{a} \in \mathbb{F}^{n}_{q}}v_{\textbf{a}}|\textbf{a}\rangle,\]
where $\{|\textbf{a}\rangle =|a_{1}\rangle\bigotimes|a_{2}\rangle\bigotimes\cdots\bigotimes|a_{n}\rangle: (a_{1}, a_{2}, \ldots, a_{n}) \in \mathbb{F}^{n}_{q}\}$ is a basis of $\mathbb{C}^{q^{n}}$ and $v_{\textbf{a}} \in \mathbb{C}$. For any two $n$-qubits $|\textbf{u}\rangle=\sum_{\textbf{a} \in \mathbb{F}^{n}_{q}}u_{\textbf{a}}|\textbf{a}\rangle$ and $|\textbf{v}\rangle=\sum_{\textbf{a} \in \mathbb{F}^{n}_{q}}v_{\textbf{a}}|\textbf{a}\rangle$, their Hermitian inner product is defined as
\[\langle\textbf{u}|\textbf{v}\rangle=\sum_{\textbf{a} \in \mathbb{F}^{n}_{q}}u_{\textbf{a}}\overline{v_{\textbf{a}}} \in \mathbb{C},\]
where $\overline{v_{\textbf{a}}}$ is the conjugate of $v_{\textbf{a}}$ in the complex field. $|\textbf{u}\rangle$ and $|\textbf{v}\rangle$ are called distinguishable if $\langle\textbf{u}|\textbf{v}\rangle=0$.

A quantum code of length $n$ is just defined as a subspace of $\mathbb{C}^{q^{n}}$. The quantum errors in a quantum system are some unitary operators. The set of error operators on $\mathbb{C}^{q^{n}}$ is defined as
\begin{eqnarray*}
  \mathcal{E}_{n} &=& \{\zeta_{p}^{i}X(\textbf{a})Z(\textbf{b}): 0 \leq i \leq p-1, \text{ where }\\
    & & \textbf{a}=(a_{1}, \cdots, a_{n}), \textbf{b}=(b_{1}, \cdots, b_{n}) \in \mathbb{F}^{n}_{q}\},
\end{eqnarray*}
where $\zeta_{p}$ is a complex primitive $p$-th root of unity. The actions of $X(\textbf{a})$ and $Z(\textbf{b})$ on the basis $|\textbf{v}\rangle \in \mathbb{C}^{q^{n}}$ ($\textbf{v} \in \mathbb{F}^{n}_{q}$) are defined as
\[X(\textbf{a})|\textbf{v}\rangle=|\textbf{v}+\textbf{a}\rangle
\textnormal{ and }
Z(\textbf{b})|\textbf{v}\rangle=\zeta_{p}^{tr(\langle \textbf{v}, \textbf{b} \rangle_{E})}|\textbf{v}\rangle,\]
respectively, where $tr(\cdot)$ is the trace function from $\mathbb{F}_{q}$ to $\mathbb{F}_{p}$. The error set $\mathcal{E}_{n}$ forms a non-abelian group and has nice property (see \cite{KKKS06}). For any error $E=\zeta_{p}^{i}X(\textbf{a})Z(\textbf{b})$, we define the quantum weight of $E$ by
\[w_{Q}(E)=\sharp\{i : (a_{i}, b_{i}) \neq (0,0)\}.\]
Let $\mathcal{E}_{n}(\ell)=\{E \in \mathcal{E}_{n}: w_{Q}(E) \leq \ell\}$ be the set of error operators with weight no more than $\ell$.
A quantum code $Q$ can detect a quantum error $E$ if and only if for any $|\textbf{u}\rangle, |\textbf{v}\rangle \in Q$ with $\langle\textbf{u}|\textbf{v}\rangle=0$, we have $\langle\textbf{u}|E|\textbf{v}\rangle=0$. The quantum code $Q$ has minimum distance $d$ if $d$ is the largest integer such that for any $|\textbf{u} \rangle, |\textbf{v}\rangle \in Q$ with $\langle\textbf{u}|\textbf{v}\rangle=0$ and $E \in \mathcal{E}_{n}(d-1)$, we have $\langle\textbf{u}|E|\textbf{v}\rangle=0$. We denote by $((n, K, d))_{q}$ or $[[n,k,d]]_{q}$ a $q$-ary quantum code of length $n$, dimension $K$ and minimum distance $d$, where $k=\log_{q}K$.

Calderbank \emph{et al.} \cite{CS96} and  Steane \cite{S96} provided an effective mathematical method to construct nice quantum codes by using
character theory of finite abelian groups. Suppose $S$ is an abelian subgroup of $\mathcal{E}_{n}$, they define the quantum stabilizer code $C(S)$ associated with $S$ to be
\[C(S)=\{|\psi\rangle \in \mathbb{C}^{q^{n}}: E|\psi\rangle=|\psi\rangle, \forall E \in S\}.\]
In other words, the quantum stabilizer code $C(S)$ is the simultaneous +1 eigenspace of all elements of $S$. Quantum stabilizer codes are analogues of classical additive codes, and classical linear codes with certain orthogonality can be used to construct quantum stabilizer codes (see \cite{CS96, S96, CRSS98, AK01}).

When the subgroup $S$ of $\mathcal{E}_{n}$ is non-abelian, the method of constructing quantum stabilizer codes does not work. This case was investigated by Brun \emph{et al.} in \cite{BDH06} by extending $S$ to be a new abelian subgroup in a larger error group. They then introduced the entanglement-assisted quantum error-correcting codes (EAQECCs), which is a generalization of quantum stabilizer codes.  It is assumed that in addition to a quantum channel, Alice (the sender) and Bob (the receiver) share a certain amount of pre-existing entangled bits (ebits), which is not subject to errors. By using the shared ebits between the sender and receiver, it is possible that the sender may send more qubits for a given number of correctable quantum errors, or correct more quantum errors for the same rate of transmission. We usually use $[[n, k, d; c]]_{q}$ to denote a $q$-ary EAQECC that encodes $k$ information qubits into $n$ channel qubits with the help of $c$ ebits (Details of the encoding procedure can be found in \cite{LB13,SHB11}), and $d$ is called the
minimum distance of the EAQECC. Such a quantum code can detect up to $d - 1$ and correct up to
$\lfloor\frac{d-1}{2}\rfloor$ quantum errors. In particular, an $[[n, k, d; c]]_{q}$ EAQECC is equivalent to a quantum stabilizer
code when $c=0$. One of the constraints among the parameters $n, k, d$ and $c$ is the following quantum Singleton bound:
\begin{lemma}(Quantum Singleton Bound \cite{LA18})\label{lem4.1}
  For any $[[n, k, d; c]]_{q}$-EAQECC, if $d \leq \frac{n+2}{2}$, we have
  \[n+c-k \geq 2(d-1).\]
\end{lemma}
An EAQECC attaining the quantum Singleton bound is called an MDS EAQECC.

 In \cite{BDH06}, Brun \emph{et al.} provided an effective mathematical method to construct $q$-ary EAQECCs by utilizing classical linear codes over finite fields without satisfying the dual containing restriction. We present their result for the Hermitian inner product as follows. For a matrix $M=(m_{ij})$ over $\mathbb{F}_{q^{2}}$, denote the conjugate transpose of $M$ by $M^{\dag}:=(m^{q}_{ji})$.

\begin{lemma} (\cite{BDH06})\label{lem4.2}
Let $H$ be the parity check matrix of an $[n, k, d]$-linear
code $\mathcal{C}$ over $\mathbb{F}_{q^{2}}$.  Then there exists an $[[n, 2k - n + c, d; c]]_{q}$ EAQECC $\mathcal{Q}$, where $c =
\textnormal{rank}(HH^{\dag})$ is the required number of maximally entangled states. In particular, if $\mathcal{C}$ is an MDS code and $d \leq \frac{n+2}{2}$, then $\mathcal{Q}$ is an MDS EAQECC.
\end{lemma}

Guenda \emph{et al.} \cite{GJG18} provided the relation between the value of $\text{rank}(HH^{\dag})$ and the dimension of the Hermitian hull of the linear code with parity check matrix $H$.

\begin{lemma} (\cite{GJG18})\label{lem4.3}
Let $\mathcal{C}$ be a $q^{2}$-ary $[n, k, d]$-linear code.
Assume that $H$ is a parity check matrix of $\mathcal{C}$. Then we have
\begin{eqnarray*}
  \textnormal{rank}(HH^{\dag}) &=& n-k-\dim(Hull_{H}(\mathcal{C})) \\
   &=& n-k-\dim(Hull_{H}(\mathcal{C}^{\perp_{H}})).
\end{eqnarray*}
\end{lemma}

Since the Hermitian dual code of an $[n, k, n-k+1]$-MDS code is an $[n, n-k, k+1]$-MDS code, we immediately obtain the following result from Lemmas \ref{lem4.1}, \ref{lem4.2} and \ref{lem4.3}.
\begin{lemma}\label{lem4.4}
Let $\mathcal{C}$ be an $[n,k]$-MDS code over $\mathbb{F}_{q^{2}}$ and $\ell =\dim(Hull_{H}(\mathcal{C}))$. If $k \leq \frac{n}{2}$, then there
exists an $[[n, n- k - \ell, k+1;  k - \ell]]_{q}$ MDS EAQECC.
\end{lemma}

\begin{remark}\label{rem4.5}
The required number $c$ of maximally entangled states of the MDS EAQECCs constructed in Lemma \ref{lem4.4} satisfies that $0 \leq c=k-\ell \leq k$. When $c=0$, then $\dim(Hull_{H}(\mathcal{C}))=\ell=k$, i.e., $\mathcal{C} \subseteq \mathcal{C}^{\perp_{H}}$. At this time, the MDS EAQECC is equivalent to an MDS quantum stabilizer code and Lemma \ref{lem4.4} is equivalent to the well-known CSS Construction (for MDS quantum stabilizer code).  When $c=k$, then $\dim(Hull_{H}(\mathcal{C}))=\ell=0$, i.e., $\mathcal{C}$ is Hermitian LCD code.  In \cite{CMTQP18}, the authors completely determined the $q$-ary Hermitian LCD codes for $q >2$.
\end{remark}

By Lemma \ref{lem4.4} and Theorem \ref{thm3.5}, we can directly construct the $q$-ary MDS EAQECCs of length $n \leq q$, which have also been obtained in \cite{GGJT18} with different approach.

\begin{theorem}\label{thm4.6}
Let $q$ be a prime power. Assume that $1 < n \leq q$. Then for any $1 \leq k \leq \frac{n}{2}$ and $0 \leq \ell \leq k$, there exists an $[[n, n- k - \ell, k+1;  k - \ell]]_{q}$ MDS EAQECC.
\end{theorem}

\begin{remark}
From Theorem \ref{thm4.6}, we have completely determined the $q$-ary MDS EAQECCs of length $n \leq q$ for all possible parameters. This result has also already been given in \cite{GGJT18} via linear codes over $\mathbb{F}_{q}$ with Euclidean inner product. Herein, we use the linear codes over $\mathbb{F}_{q^{2}}$ with Hermitian inner product.
\end{remark}

Similarly, from Theorem \ref{thm3.6} and Theorems \ref{thm3.8}-\ref{thm3.11}, we obtain the following six families of MDS EAQECCs with flexible parameters.

\begin{theorem}\label{thm4.8}
Suppose  $q=p^{m} \geq 3$ and $r=p^{e}$, where $e \mid m$. Let $n=tr^{z}$, where $1 \leq t \leq r$ and $1 \leq z \leq 2\frac{m}{e}-1$.
Then for any $1 \leq k \leq \lfloor\frac{n-1+q}{q+1}\rfloor$,
  \begin{description}
    \item[(i)] if $0 \leq \ell \leq k$, there exists an $[[n, n- k - \ell, k+1;  k - \ell]]_{q}$ MDS EAQECC;
    \item[(ii)] if $0 \leq \ell \leq k-1$, there exists an $[[n+1, n+1-k- \ell, k+1;  k - \ell]]_{q}$ MDS EAQECC.
  \end{description}
\end{theorem}

\begin{theorem}\label{thm4.9}
Let $q>2$ be a prime power and $n' \mid (q^{2}-1)$. Let $n=tn'$, where $1 \leq t \leq \frac{q-1}{n_{1}}$ and $n_{1}=\frac{n'}{\gcd(n',q+1)}$. Then for any $1 \leq k \leq \lfloor\frac{n+q}{q+1}\rfloor$,
\begin{description}
  \item[(i)] if $0 \leq \ell \leq k-1$, there exists an $[[n, n- k - \ell, k+1;  k - \ell]]_{q}$ MDS EAQECC;
  \item[(ii)]if $0 \leq \ell \leq k$, there exists an $[[n+1, n+1- k - \ell, k+1;  k - \ell]]_{q}$ MDS EAQECC;
  \item[(iii)] if $0 \leq \ell \leq k-1$, there exists an $[[n+2, n+2- k - \ell, k+1;  k - \ell]]_{q}$ MDS EAQECC.
\end{description}
\end{theorem}

\begin{theorem}\label{thm4.10}
Let $q>2$ be a prime power. Then for any $0 \leq \ell \leq q$, there exists a $[[q^{2}+1, q^{2}+1- q - \ell, q+1;  q - \ell]]_{q}$ MDS EAQECC.
\end{theorem}

\begin{remark}\label{rem4.11}
In \cite{QZ18, FCX16, LLGML18, LZLK19, CZK18, LLLM18, K19, LC19}, the authors provided several constructions of MDS EAQECCs with fixed required number of maximally entangled states. In \cite{LCC} and \cite{GGJT18}, the authors constructed several families of  $q$-ary MDS EAQECCs with length less than or equal to $q+1$ and the required number of maximally entangled states can take all or almost all possible values. In our Theorem \ref{thm4.6} and Theorems \ref{thm4.8}-\ref{thm4.10}, we provide several classes of MDS EAQECCs with flexible parameters. Moreover, the lengths of these $q$-ary EAQECCs can be larger than $q+1$ and the required number of maximally entangled states can also take arbitrarily possible values. Hence many new MDS EAQECCs are obtained.
\end{remark}

In the following, in order to illustrate our results obtained in Theorems \ref{thm4.8}-\ref{thm4.10}, we list some examples of $q$-ary MDS EAQECCs with length larger than $q+1$ in Tables 1-3.

\begin{table}[htbp]
\centering
\caption{Examples of MDS EAQECCs of Theorem \ref{thm4.8} (\textnormal{i}) for $q=9$ and $n=72$ ($t=8, r=9, z=1$)}
\begin{tabular}{ccc|ccc}
  \hline
  $k$ &$\ell$& MDS EAQECCs &$k$ &$\ell$& MDS EAQECCs \\
  \hline
  3 &1& $[[72,68,4;2]]_{9}$ &6 &2& $[[72,64,7;4]]_{9}$\\
  3 &2& $[[72,67,4;1]]_{9}$ &6 &3& $[[72,63,7;3]]_{9}$\\
  4 &1& $[[72,67,5;3]]_{9}$ &6 &4& $[[72,62,7;2]]_{9}$\\
  4 &2& $[[72,66,5;2]]_{9}$ &6 &5& $[[72,61,7;1]]_{9}$\\
  4 &3& $[[72,65,5;1]]_{9}$ &7 &1& $[[72,64,8;6]]_{9}$\\
  5 &1& $[[72,66,6;4]]_{9}$ &7 &2& $[[72,63,8;5]]_{9}$\\
  5 &2& $[[72,65,6;3]]_{9}$ &7 &3& $[[72,62,8;4]]_{9}$\\
  5 &3& $[[72,64,6;2]]_{9}$ &7 &4& $[[72,61,8;3]]_{9}$\\
  5 &4& $[[72,63,6;1]]_{9}$ &7 &5& $[[72,60,8;2]]_{9}$\\
  6 &1& $[[72,65,7;5]]_{9}$ &7 &6& $[[72,59,8;1]]_{9}$ \\
  \hline
\end{tabular}
\end{table}

\begin{table}[htbp]
\centering
\caption{Examples of MDS EAQECCs of Theorem \ref{thm4.9} (\textnormal{i}) for $q=11$ and $n=96$ ($t=8, n'=12$)}
\begin{tabular}{ccc|ccc}
  \hline
  $k$ &$\ell$& MDS EAQECCs & $k$ &$\ell$& MDS EAQECCs \\
  \hline
  2 &1& $[[96,93,3;1]]_{11}$ & 6 &4& $[[96,86,7;2]]_{11}$ \\
  3 &1& $[[96,92,4;2]]_{11}$ & 6 &5& $[[96,85,7;1]]_{11}$ \\
  3 &2& $[[96,91,4;1]]_{11}$ & 7 &1& $[[96,88,8;6]]_{11}$ \\
  4 &1& $[[96,91,5;3]]_{11}$ & 7 &2& $[[96,87,8;5]]_{11}$ \\
  4 &2& $[[96,90,5;2]]_{11}$ & 7 &3& $[[96,86,8;4]]_{11}$ \\
  4 &3& $[[96,89,5;1]]_{11}$ & 7 &4& $[[96,85,8;3]]_{11}$ \\
  5 &1& $[[96,90,6;4]]_{11}$ & 8 &1& $[[96,87,9;7]]_{11}$ \\
  5 &2& $[[96,89,6;3]]_{11}$ & 8 &2& $[[96,86,9;6]]_{11}$ \\
  5 &3& $[[96,88,6;2]]_{11}$ & 8 &3& $[[96,85,9;5]]_{11}$ \\
  5 &4& $[[96,87,6;1]]_{11}$ & 8 &4& $[[96,84,9;4]]_{11}$ \\
  6 &1& $[[96,89,7;5]]_{11}$ & 8 &5& $[[96,83,9;3]]_{11}$ \\
  6 &2& $[[96,88,7;4]]_{11}$ & 8 &6& $[[96,82,9;2]]_{11}$ \\
  6 &3& $[[96,87,7;3]]_{11}$ & 8 &7& $[[96,81,9;1]]_{11}$ \\
  \hline
\end{tabular}
\end{table}

\begin{table}[htbp]
\centering
\caption{Examples of MDS EAQECCs of Theorem \ref{thm4.10} for $q=5,7$ and 13}
\begin{tabular}{ccc|ccc}
  \hline
  $q$ &$\ell$& MDS EAQECCs & $q$ &$\ell$& MDS EAQECCs\\
  \hline
  5 &1& $[[26,20,6;4]]_{5}$ &13 &2& $[[170,151,14;11]]_{13}$\\
  5 &2& $[[26,19,6;3]]_{5}$ &13 &3& $[[170,150,14;10]]_{13}$\\
  5 &3& $[[26,18,6;3]]_{5}$ &13 &4& $[[170,149,14;9]]_{13}$\\
  5 &4& $[[26,17,6;1]]_{5}$ &13 &5& $[[170,148,14;8]]_{13}$\\
  7 &1& $[[50,42,8;6]]_{7}$ &13 &6& $[[170,147,14;7]]_{13}$\\
  7 &2& $[[50,41,8;5]]_{7}$ &13 &7& $[[170,146,14;6]]_{13}$\\
  7 &3& $[[50,40,8;4]]_{7}$ &13 &8& $[[170,145,14;5]]_{13}$\\
  7 &4& $[[50,39,8;3]]_{7}$ &13 &9& $[[170,144,14;4]]_{13}$\\
  7 &5& $[[50,38,8;2]]_{7}$ &13 &10& $[[170,143,14;3]]_{13}$\\
  7 &6& $[[50,37,8;1]]_{7}$ &13 &11& $[[170,142,14;2]]_{13}$\\
  13 &1& $[[170,152,14;12]]_{13}$ &13 &12& $[[170,141,14;1]]_{13}$\\
  \hline
\end{tabular}
\end{table}

\begin{remark}
According to Remark \ref{rem4.11}, we do not list the MDS EAQECCs with $\ell=0$ and $\ell=k$ in Tables 1-3.
\end{remark}
\section{Conclusion}
In this paper, we study the hull of a linear code with both Euclidean and Hermitian inner products. We employ some additive subgroups of the finite field $\mathbb{F}_{q}$ (or $\mathbb{F}_{q^{2}}$) and multiplicative subgroups of $\mathbb{F}^{*}_{q}$ (or $\mathbb{F}^{*}_{q^{2}}$) and their cosets to construct the desired GRS codes and extended GRS codes. Then we can determine the dimensions of the Euclidean or Hermitian hulls of these codes. Several families of MDS codes with  Euclidean or Hermitian hulls of arbitrary dimensions were thus obtained. Finally, we apply these results to construct several new classes of MDS EAQECCs with flexible parameters.  In particular, several classes of $q$-ary MDS EAQECCs with length $n >q+1$ are also constructed. Note that in Theorem \ref{thm3.6} and Theorems \ref{thm3.8}-\ref{thm3.11}, the dimension $k$ of the $q^{2}$-ary MDS codes of length $n$ is roughly bounded by $\lfloor\frac{n+q}{q+1}\rfloor$. Therefore, constructing suitable $q^{2}$-ary MDS codes of larger dimension $\lfloor\frac{n+q}{q+1}\rfloor < k \leq \frac{n}{2}$ and determining the dimensions of their Hermitian hulls will be one of research directions in our future work.

\vskip 5mm \noindent {\bf Acknowledgments} We sincerely thank Professor Xiwang Cao for his helpful suggestions and comments.

\end{document}